\documentclass{llncs}
\makeatletter
\let\proof\@undefined
\let\endproof\@undefined
\makeatother

\usepackage{enumerate}
\usepackage{url}
\usepackage{ulem}
\usepackage{subfigure}
\usepackage{a4wide}
\usepackage{color}
\usepackage{marginnote}

\usepackage{amsfonts}
\usepackage{amsthm}

   \usepackage[pdftex]{graphicx}
\usepackage{amsmath}

\usepackage{color}

\newcommand{\cX}{{\mathcal X}}

\newcommand{\cC}{{\mathcal C}}
\newcommand{\cT}{{\mathcal T}}

\newcommand{\contains}{{\rightarrow_{\cC}}}

\newcommand{\notcontains}{{\not \rightarrow}_{\cC}}

\newtheorem{observation}{Observation}

\begin{document}
\title{A note on efficient computation of hybridization number via softwired clusters}
\author{Steven Kelk}
\institute{Department of Knowledge Engineering (DKE), Maastricht University\\ 
P.O. Box 616, 6200 MD Maastricht, The Netherlands. steven.kelk@maastrichtuniversity.nl}
\maketitle

\begin{abstract}
Here we present a new fixed parameter tractable algorithm to compute the
hybridization number $r$ of two rooted binary phylogenetic trees on taxon set
$\cX$ in time $(6r)^r \cdot poly(n)$, where $n=|\cX|$. The novelty of this
approach is that it avoids the use of Maximum Acyclic Agreement Forests (MAAFs) and
instead exploits the equivalence of the problem with a related problem from
the \emph{softwired clusters} literature. This offers an alternative perspective on the
underlying combinatorial structure of the hybridization number problem.
\end{abstract}

\section{Introduction}

For notation and background we refer the reader to \cite{elusiveness,fptclusters}. Let $\cT = \{ T_1, T_2 \}$ be a set of two rooted binary phylogenetic trees on $\cX$, where $|\cX|=n$. We may assume without loss of generality that $T_1$ and $T_2$ do not have any non-trivial common subtrees, so $|\cX| \geq 3$. Let $\cC = Cl(\cT) = Cl(T_1) \cup Cl(T_2)$ be the set of clusters obtained from $\cT$. Clearly, every taxon in $\cX$ appears in at least one cluster in $\cC$, and
$|\cC| \leq 4(n-1)$, because a binary tree on $n$ taxa contains exactly $2(n-1)$ edges.

Now, we know from \cite{twotrees} that $r(\cC) = h(\cT)$ where $h(\cT)$ is the
hybridization number of the two trees $T_1$ and $T_2$. That is, the minimum number of reticulations required to display the two trees, $h(\cT)$, is equal to the
minimum number of reticulations required to represent the union of the clusters obtained from the two trees, $r(\cC)$. Hence we can concentrate on
computing $r(\cC)$. In \cite{elusiveness,twotrees} it is also proven that there exists a \emph{binary} network $N$ that represents $\cC$ such that $r(\cC)=r(N)$,
so we can further restrict our attention to binary networks.

 In \cite{fptclusters} an algorithm with running time $f( r(\cC) ) \cdot poly(n)$ is given to compute $r(\cC)$ for an \emph{arbitrary} set of clusters,
where $f( r(\cC) )$ is a function of $r(\cC)$ that does not depend on $n$ and $poly(n)$ is a function of the form $n^{O(1)}$. However, the running time of the algorithm in \cite{fptclusters} is purely theoretical. In the case of clusters obtained from two binary trees the running time can be more heavily optimized, which is the motivation for this note.

\section{A fixed parameter tractable algorithm for computing the hybridization number of two binary trees}

Given a cluster set $\cC$ and $x,y \in \cX$, we write $x \rightarrow_{\cC}  y$ if and only if every non-singleton cluster in $\cC$ containing $x$, also contains $y$\footnote{Note that, if a taxon $x$ appears
in only one cluster, $\{x\}$, then (vacuously) $x \contains y$ for all $y \neq x$.}. We say that a taxon $x \in \cX$ is a \emph{terminal} if there does not exist $x' \in \cX$ such that $x \neq x'$ and
$x \rightarrow_{\cC} x'$.  In \cite{elusiveness} it is proven that $\emptyset \subset \cX' \subset \cX$ is an ST-set of $\cC$ if and only if $\cX' \in Cl(T_1) \cap Cl(T_2)$ and the two subtrees induced by $\cX'$ in $T_1$ and $T_2$ are identical. Given that $T_1$ and $T_2$ are assumed to have no non-trivial common subtrees it follows that
$\cC$ has no non-singleton ST-sets, a property we call ST-collapsed \cite{fptclusters}. 

\begin{observation}
\label{obs:poset}
Let $\cC$ be an ST-collapsed set of clusters on $\cX$ such that $r(\cC) \geq 1$. Then the relation $\contains$ is a partial order on $\cX$, the terminals are the maximal
elements of the partial order and there is at least one terminal.
\end{observation}
\begin{proof}
The relation $\contains$ is clearly reflexive and transitive. To see that it is anti-symmetric, suppose there exist two elements $x \neq y \in \cX$ such that $x \contains y$
and $y \contains x$. Then we have that, for every non-singleton cluster $C \in \cC$, $C \cap \{x,y\}$ is either equal to $\emptyset$ or $\{x,y\}$ i.e. $C$ is compatible
with $\{x,y\}$. Furthermore, the only clusters that can possibly be in $\cC | \{x,y\}$ are $\{x\}, \{y\}$ and $\{x,y\}$ and these are all mutually compatible. So
 $\{x,y\}$ is an ST-set, contradicting the fact that $\cC$ is ST-collapsed. Hence $\contains$ is a partial order. The fact that the terminals are the maximal elements
of the partial order then follows immediately from their definition. Finally, it is well known that every partial order on a finite set of elements contains at least one maximal element (because otherwise a cycle exists which contradicts the anti-symmetry property).
\end{proof}


Let $T$ be a phylogenetic tree on $\cX$. For a vertex $u$ of $T$ we define $\cX(u) \subseteq \cX$ to be the set of all taxa that can be reached from $u$ by directed
paths. For a taxon $x \in \cX$ we define $W^{T}(x)$, the \emph{witness set} for $x$ in $T$, as $\cX(u) \setminus \{x\}$, where $u$ is the parent of $x$. A critical
property of $W^{T}(x)$ is that, for any non-singleton cluster $C \in Cl(T)$ that contains $x$, $W^{T}(x) \subseteq C$ \cite{elusiveness}.

\begin{observation}
\label{obs:termequiv}
Let $\cC = Cl(\cT)$ be a set of clusters on $\cX$, where $\cT = \{T_1,T_2\}$ is a set of two binary trees on $\cX$ with no
non-trivial common subtrees, and $r(\cC) \geq 1$. Then for any $x \in \cX$ the following statements are equivalent:
(1) $x$ is a terminal of $\cC$;
(2) there exist incompatible clusters $C_1, C_2 \in \cC$ such that $C_1 \cap C_2 = \{x\}$;
(3) $W^{T_1}(x) \cap W^{T_2}(x) = \emptyset$.
\end{observation}
\begin{proof}
We first prove that (2) implies (1). For $x' \not \in C_1 \cup C_2$ it holds that $x \notcontains x'$, because
$x \in C_1$ but $x' \not \in C_ 1$. For $x' \in C_1 \setminus C_2$ it cannot hold that $x \contains x'$, because
$x \in C_2$ but $x' \not \in C_2$, and this holds symmetrically for $x' \in C_2 \setminus C_1$. Hence $x$ is
a terminal. We now show that (1) implies (3).
Suppose (3) does not hold. Then there exists some taxon $x' \in W^{T_1}(x) \cap W^{T_2}(x)$. So
every non-singleton cluster in $\cC$ that contains $x$ also contains $x'$, irrespective of whether the cluster came from $T_1$ or $T_2$.
But then $x \contains x'$, so (1) does not hold. Hence (1) implies (3). Finally, we show that (3) implies (2). Note that (3) implies
that in both $T_1$ and $T_2$ the parent of $x$ is \emph{not} the root. If this was not so, then (wlog) $W^{T_1}(x) = \cX \setminus \{x\}$,
and combining this with the fact that $W^{T_1}(x), W^{T_2}(x) \neq \emptyset$ would contradict (3). 
Hence $W^{T_1}(x) \cup \{x\} \in Cl(T_1)$ and $W^{T_2}(x) \cup \{x\} \in Cl(T_2)$, from which (2) follows.
\end{proof}


\begin{observation}
\label{obs:containswitness}
Let $\cC = Cl(\cT)$ be a set of clusters on $\cX$, where $\cT = \{T_1,T_2\}$ is a set of two binary trees on $\cX$ with no
non-trivial common subtrees, and $r(\cC) \geq 1$. Then, for any two taxa $x \neq y \in \cX$, $x \contains y$ if and only
if $y \in W^{T_1}(x) \cap W^{T_2}(x)$.
\end{observation}
\begin{proof}
Suppose $x \contains y$, but (wlog) $y \not \in W^{T_1}(x)$. The parent of $x$ in $T_1$ cannot be the root, because
then $W^{T_1}(x) = \cX \setminus \{x\}$ which contains $y$. So there is an edge in $T_1$ whose head is the parent of
$x$. Let $C \in \cC$ be the cluster represented by this edge, then $C = \{x\} \cup W^{T_1}(x)$. $W^{T_1}(x)$ is non-empty and contains
neither $x$ nor $y$, so $C$ is a non-singleton cluster which contains $x$ but not $y$. So $x \notcontains y$.  
In the other direction, suppose $y \in W^{T_1}(x) \cap W^{T_2}(x)$. Let $C \in \cC$ be a non-singleton cluster that contains $x$. Every non-singleton cluster $C \in \cC$ is from $Cl(T_1)$ or $Cl(T_2)$, so $W^{T_1}(x) \subseteq C$ or $W^{T_2}(x) \subseteq C$. In any case it follows that $y \in C$.
\end{proof}


\begin{lemma}
\label{lem:term2term}
Let $\cC = Cl(\cT)$ be a set of clusters on $\cX$, where $\cT = \{T_1,T_2\}$ is a set of two binary trees on $\cX$ with no
non-trivial common subtrees, and $r(\cC) \geq 1$. Let $x$ be any taxon in $\cC$. If $x$ is not a terminal of $\cC$ then
there exists a terminal $y$ such that $x \contains y$.
\end{lemma}
\begin{proof}
$\cC$ is ST-collapsed because $T_1$ and $T_2$ contain no non-trivial common subtrees. Hence, by Observation
\ref{obs:poset}, we know that $\contains$ is a partial order on $\cX$ and the terminals, of which there is at least one, are the maximal elements
of the partial order. The result then follows immediately from the transitivity and anti-symmetry property of partial orders.
\end{proof}

\begin{lemma} 
\label{lem:SBRexists}
Let $\cC = Cl(\cT)$ be a set of clusters on $\cX$, where $\cT = \{T_1,T_2\}$ is a set of two binary trees on $\cX$ with no
non-trivial common subtrees, and $r(\cC) \geq 1$. Then there exists $x \in \cX$ such that $r( \cC \setminus \{x\} ) < r(\cC)$.
\end{lemma}
\begin{proof}
Consider a binary network $N$ which represents $\cC$, where $r(N)=r(\cC)$.
By acyclicity $N$ contains at least one \emph{Subtree Below a Reticulation} (SBR) \cite{elusiveness},
i.e. a node $u$ with indegree-1 whose parent is a reticulation, and such that no reticulation can be reached by a directed
path from $u$.  Let $\cX'$ be the set of taxa reachable from $u$ by directed paths. $\cX'$ is an ST-set, so $|\cX'|=1$ (because
$\cC$ is ST-collapsed). Let $x$ be the single taxon in $\cX'$. Deleting $x$ and its reticulation parent from $N$ (and tidying up
the resulting network in the usual fashion) creates a network $N'$ on $\cX \setminus \{x\}$ with $r(N') < r(N)$ that represents $\cC \setminus \{x\}$.
\end{proof}

\begin{lemma} 
\label{lem:atmostone}
Let $\cC = Cl(\cT)$ be a set of clusters on $\cX$, where $\cT = \{T_1,T_2\}$ is a set of two binary trees on $\cX$ with no
non-trivial common subtrees, and $r(\cC) \geq 1$. Then for each $x \in \cX$ it holds that $r(\cC) -1 \leq r( \cC \setminus \{x\} ) \leq r(\cC)$.
\end{lemma}
\begin{proof}
The second $\leq$ is immediate because removing a taxon from a cluster set cannot raise the reticulation number of the cluster set. The first $\leq$ holds
because in \cite{elusiveness} it is shown how, given \emph{any} network $N'$ on $\cX \setminus \{x\}$ that represents $\cC \setminus \{x\}$,
we can extend $N'$ to obtain a network $N$ on $\cX$ that represents $\cC$ such that $r(N) \leq r(N')+ 1$.
\end{proof}

Recall from \cite{elusiveness} the definition of an \emph{ST-set tree sequence} of a cluster set $\cC$. Let $(S_1, S_2, ..., S_p)$ be
an ST-set tree sequence of $\cC$ of minimum length, where $\cC = Cl(\cT) = Cl(T_1) \cup Cl(T_2)$. In \cite{elusiveness} it is proven that $p = r(\cC)$. 

\begin{lemma} 
\label{lem:cherries}
Let $\cC = Cl(\cT)$ be a set of clusters on $\cX$, where $\cT = \{T_1,T_2\}$ is a set of two binary trees on $\cX$ with no
non-trivial common subtrees, and $r(\cC) \geq 1$. Suppose there exist distinct taxa $a,b,c \in \cX$ such that
$\{a,b\}$ and $\{b,c\}$ are both clusters in $\cC$. Then there exists $x \in \{a,b,c\}$ such that
$r(\cC \setminus \{x\}) = r(\cC) - 1$.
\end{lemma}
\begin{proof}
We know from \cite{elusiveness} that there exists an ST-set tree sequence $(S_1, ..., S_p)$ of $\cC$ where $p = r(\cC)$. It is clear that at least one
of $\{a,b,c\}$ has to occur in one of the $S_i$ sets, because otherwise the two clusters $\{a,b\}$ and $\{b,c\}$
survive even after $S_p$ has been removed, contradicting the fact that it is a \emph{tree} sequence. Now, let $1 \leq i \leq p$ be
the smallest $i$ such that $S_i \cap \{a,b,c\} \neq \emptyset$. Note that, at the point just before the ST-set $S_i$ is
removed, none of $\{a,b\}, \{b,c\}, \{a,c\}, \{a,b,c\}$ are ST-sets (because of the incompatible pair
of clusters $\{a,b\}$ and $\{b,c\}$). So $|S_i \cap \{a,b,c\}| < 3$. Furthermore, $|S_i \cap \{a,b,c\}| \neq 2$
because at least one of the two clusters $\{a,b\}$ and $\{b,c\}$ will be incompatible with $S_i$. So
$|S_i|=1$ and $S_i \subseteq \{a,b,c\}$. Let $x$ be the single taxon in $S_i$. This means that $(S_1, ..., S_{i-1}, S_{i+1}, ..., S_{p})$ is
an ST-set tree sequence of $\cC \setminus \{x\}$ of length $p-1$. We conclude that $r(\cC \setminus \{x\}) < r(\cC)$ and (by Lemma \ref{lem:atmostone})
we have that $r(\cC \setminus \{x\}) = r(\cC)-1$.
\end{proof}

\begin{theorem}
\label{thm:6rA}
Let $\cC = Cl(\cT)$ be a set of clusters on $\cX$, where $\cT = \{T_1,T_2\}$ is a set of two binary trees on $\cX$ with no
non-trivial common subtrees, and $r(\cC) \geq 1$. Then at least one of the following two conditions holds: (1) there exist distinct taxa $a,b,c \in \cX$
such that $\{a,b\}$ and $\{b,c\}$ are both clusters in $\cC$; (2) there exists a taxon $x \in \cX$ such that
$r(\cC \setminus \{x\}) = r(\cC) - 1$ and $x$ is either a terminal of $\cC$ or is in some size-2 cluster with a terminal of $\cC$.
\end{theorem}
\begin{proof}
Towards a counter-example, assume that the claim does not hold for $\cC$. For each $x \in \cX$ that is not a terminal,
let $M(x)$ be an arbitrary terminal such that $x \contains M(x)$, which must exist by Lemma \ref{lem:term2term}. The mapping
$M$ will remain unchanged for the rest of the proof.

Let $N$ be an arbitrary binary network that represents $\cC$ such that $r(N)=r(\cC)$. Due to the fact that $\cC$ is ST-collapsed,
every SBR consists of exactly one taxon, so we can unambiguously identify each SBR by its corresponding taxon. Let $R(N) \subseteq \cX$ be
the set of SBRs of $N$. Clearly, if $R(N)$ contains a terminal we are done (by the same argument as used in the proof of Lemma \ref{lem:SBRexists}).

For $x \in R(N)$, we define the \emph{detach and re-hang above a terminal} (DRHT) operation as follows. We first delete
$x$ and tidy up the resulting network in the usual fashion, which causes the parent reticulation of $x$ and its reticulation edges to disappear. This creates a network $N'$ on $\cX \setminus \{x\}$ that represents $\cC \setminus \{x\}$, where
$r(N')=r(N)-1$ (by Lemma \ref{lem:atmostone}).\footnote{Note that Lemma \ref{lem:atmostone} also prevents that, prior to the tidying-up phase, a multi-edge is
created, because then both the parent reticulation of $x$ \emph{and} the reticulation at the end of the multi-edge would disappear in the tidying-up phase, meaning
that $r(N') \leq r(N)-2$.} Let $\delta$ be the root of $N'$ and $p$ be the parent of $M(x)$ in $N'$. We construct a new network
$N''$ from $N'$ by deleting the edge $(p,M(x))$, introducing new nodes $p'$, $r$, $r'$, adding the edges
$(p,p'), (p',M(x)), (p',r), (\delta,r), (r,r')$ and finally labelling $r'$ with $x$. This potentially raises the outdegree of the
root above 2 (i.e. makes the network non-binary) but this could easily be addressed by replacing the root with a chain
of nodes of indegree at most 1 and outdegree 2 (see e.g. \cite{twotrees}); the exposition is easier to follow if we permit a high-degree root.
As observed in \cite{elusiveness} $N''$ represents $\cC$. To summarise the argument from \cite{elusiveness}, let $C \in \cC$ be a non-singleton cluster such that $x \not \in C$. Clearly $C$ is represented by $N'$. The edge in $N'$ that represents $C$ will still represent $C$ in $N''$ if we switch
the reticulation edge $(p',r)$  \emph{off} and the reticulation edge $(\delta,r)$ \emph{on}. So suppose $C$ is a non-singleton
cluster that does contain $x$. In this case $M(x) \in C \setminus \{x\}$. So the edge in $N'$ that represents
$C \setminus \{x\}$ can represent $C$ in $N''$ by switching the reticulation edge $(p',r)$  \emph{on} and the reticulation edge $(\delta,r)$ \emph{off}.
Hence $N''$ represents all clusters $C \in \cC$ and $r(N'')=r(N)$.
Note that, in general, $R(N'') \neq R(N)$.

We will repeatedly apply the DRHT operation to transform $N$ into a network with a ``canonical'' form. Specifically, choose an arbitrary $x \in R(N)$ and let
$R^{0} = \{x\}$. Let $N^{0}=N$ and let $N^{1}$ be the network obtained by applying the DRHT operation to $x$ in $N$. We apply the following
procedure, starting with $i=0$:\\
\\
(1) If $R(N^{i+1})$ contains a terminal then stop.\\
(2) If $R(N^{i+1}) \setminus R^{i} = \emptyset$ then stop.\\
(3) Otherwise, let $y$ be an arbitrary taxon in $R(N^{i+1}) \setminus R^{i}$, let $R^{i+1} = R^{i} \cup \{y\}$ and let $N^{i+2}$ be obtained by applying
the DRHT operation to $y$ in $N^{i+1}$. Increment $i$ and go back to (1).\\
\\
The procedure will definitely stop because with each new iteration $|R^{i+1}| > |R^{i}|$. Let $N^{*}$ be the final network obtained.
Note that if the procedure stops at line (1) the proof
is complete: $N^{*}$ is a network that represents $\cC$ with $r(\cC)$ reticulations that has a terminal as a SBR. So we can assume
that the procedure stops at line (2). Clearly, there are no terminals in $R(N^{*})$. Furthermore, for each $x \in R(N^{*})$ we
know that at some iteration a DRHT operation was performed on $x$, because otherwise the procedure could continue for at least one more iteration. In the
iteration when this happens, the tail of one reticulation edge (of the parent reticulation of $x$)
is attached to the root, and the other ``just above'' $M(x)$ i.e. at $p'$. In subsequent iterations $M(x)$ will never undergo
a DRHT operation, because it is a terminal. Also, $x$ is not a terminal (and thus is not in the range of $M$) so the edge between $x$ and its parent reticulation
will never be subdivided by a DRHT operation. Furthermore, both reticulation edges will remain intact because, after the tidying-up phase, the DRHT operation only subdivides
edges whose head is a taxon. In fact, the only relevant change that can happen is
that the edge $(p',M(x))$ is subdivided by later DRHT operations; specifically, DRHT operations applied to some non-terminal $y \neq x$ such that $M(x)=M(y)$.
Whether or not this happens, it follows that in $N^{*}$ for every $x \in R(N^{*})$ one parent of the (parent reticulation of) $x$ is the root,
and the other parent is a node $t(x)$ such that a directed tree-path (i.e. a path that contains no reticulations) exists from $t(x)$ to $M(x)$.\\
\\
We will continue to focus on $N^{*}$. We say that a directed path is a \emph{root-reticulation} path if it starts at the root and terminates
at a reticulation. The \emph{reticulation length} of such a path is the number of reticulations in it (including the end node). Observe that the last
node $r$ on a root-reticulation path of maximum length must be (the parent reticulation of) an SBR. If this was not so then a previously unvisited reticulation $r^{*} \neq r$ is
reachable by a directed path from $r$, thus contradicting the assumption that the root-reticulation path had maximum reticulation length.

Consider an arbitrary $x' \in R(N^{*})$ corresponding to the SBR at the end of a root-reticulation path of maximum reticulation length. Let $x$ be the taxon in
$R(N^{*})$ such that $M(x)=M(x')$ and, amongst all such taxa, \emph{most recently} underwent a DRHT operation; it might be that $x=x'$. Note that, by construction, $x$ is
also at the end of a root-reticulation path of maximum reticulation length. Furthermore, one parent of (the parent reticulation of) $x$ is the root, and the other is a node $t(x)$  such that $(t(x), M(x))$ is an edge in $N^{*}$.  We are now ready for the core argument in the proof. We walk backwards from $t(x)$ towards the root until we encounter a vertex $u$ for
which one of the following three mutually exclusive cases holds: (a) $u$ is a reticulation; (b)
$u$ is a tree-node (i.e. a node that is not a reticulation) from which some taxon $y \in \cX \setminus M(x)$ can be reached by a directed tree-path; (c) $u$ is the root and (b) does not hold.\\ 
\\
Before commencing with the case analysis we argue that $N^{*}$ has a specific form.
Let $t$ be an intermediate node on the directed path from $u$ to $t(x)$. We know
$t$ is a tree-node from which no taxon (other than $M(x)$) can be reached by a directed tree-path. So all maximal directed tree-paths starting at $t$ that
do not terminate at $M(x)$, must terminate at the parent of a reticulation $r$. But then there exists a root-reticulation path of maximum reticulation length terminating
at $r$, so $r$ is actually the parent of an SBR. Let $x'$ be the taxon corresponding to this SBR. We know (by construction of $N^{*}$) that the non-root parent of $r$ can
reach $M(x')$ by a directed tree-path. If $M(x') \neq M(x)$ then (b) would actually have held for node $t$, because merging the two directed tree-paths would
create a directed tree-path from $t$ to $M(x')$, and the backwards walk would have terminated
earlier. So $M(x)=M(x')$ and hence $(t,r)$ is an edge in $N^{*}$. From this we can conclude that, for each intermediate node $t$, the child of $t$ that does
not lie on the path from $u$ to $t(x)$, is (the parent reticulation of) an SBR, and moreover all such SBRs map to $M(x)$. We now commence the case analysis.\\
\\
\textbf{Case (a)}. In this case $M(x)$ is the only taxon reachable
by directed tree-paths from the child of reticulation $u$. This is depicted in Figure \ref{fig:hitret}. Consider the network $N^{**}$ obtained from $N^{*}$ by deleting
$M(x)$ and suppressing its parent; in particular, consider how this changes Figure \ref{fig:hitret}. Clearly, $N^{**}$ represents $\cC \setminus \{M(x)\}$. Now, there
is some ST-set tree sequence of $\cC \setminus \{M(x)\}$ that begins $( \{x\}, \{x'\}, \ldots, \{x''\}$, because singletons are always ST-sets. Each time one of these
ST-sets is removed from $N^{**}$ the reticulation number of $N^{**}$ drops by exactly one, \emph{except} at the point when $\{x''\}$ is removed, because
at this point the reticulation $u$ will also disappear, causing the reticulation number of the network to drop by two. Hence we can conclude that
$\cC \setminus \{M(x)\}$ has an ST-set tree sequence of length $r(N^{*})-1 = r(\cC)-1$, and we are done.\\ 
\\
\textbf{Case (b)}. Let $y \neq M(x)$ be the taxon that can be reached by a directed tree-path from $u$. This is depicted in Figure \ref{fig:hittaxon}.
$M(x)$ is a terminal, so there must exist some non-singleton cluster $C \in \cC$ such that $M(x)$ is in $C$, but $y$ is \emph{not} in $C$. Critically, the only edges that can represent
such a cluster lie on the path from $u$ to $t(x)$.  Hence there exists $R' \subseteq R(N^{*})$ such that $C = R' \cup M(x)$ and for all $x' \in R'$, $M(x')=M(x)$. So for each $x' \in R'$ we know that
$x' \contains M(x)$. Hence in both $T_1$ and $T_2$ the parent of $M(x)$ must be reachable by a directed path
from the parent of $x'$ (possibly of length 0). Suppose there exist $y', z' \in \cX$ such that $y'  \in W^{T_1}( M(x) )$, $z' \in W^{T_2}( M(x) )$
and neither $y'$ nor $z'$ is in $R'$. But then every non-singleton cluster in $\cC$ that contains $M(x)$, contains either $y'$ or $z'$. Hence
$C \not \in \cC$, which is obviously not possible. So there must be some element $x''$ of $R'$ that appears in (wlog) $W^{T_1}(M(x))$. But there must
also be a directed path from the parent of $x''$ in $T_1$ to the parent of $M(x)$ in $T_1$. So $x''$ must be a sibling
of $M(x)$, i.e. $\{ x'', M(x) \} \in \cC$. Furthermore, $x''$ is an SBR, and $M(x)$ is a terminal, so we are done.\\
\\
\textbf{Case (c)}. In this case the network must look like Figure \ref{fig:hitroot}, because the maximum reticulation length of a root-reticulation path is 1.
$M(x)$ is in at least one non-singleton cluster $C$ (otherwise it would not be a terminal) so there again exists  $R' \subseteq R(N^{*})$ such that $C = R' \cup M(x)$ and for all $x' \in R'$, $M(x')=M(x)$. (In this case $R(N^{*}) = \cX \setminus \{M(x)\}$).
The rest of the analysis is the same as case (b).
\end{proof}
\begin{figure}
\begin{center}
\includegraphics[scale=0.2]{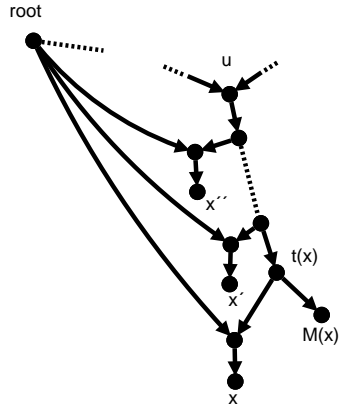}
\end{center}
\caption{This is case (a) in the proof of Theorem \ref{thm:6rA}. Here $u$ is a reticulation and the only taxon reachable by
a directed tree-path from the child of $u$ is $M(x)$. Each intermediate node on the path from $u$ to $M(x)$ is the tail of a reticulation edge that feeds
into (the parent reticulation of) an SBR; the other reticulation edge is attached to the root. All these SBRs $x$, $x'$,...,$x''$ are such that $M(x)=M(x')=...=M(x'')$.} 
\label{fig:hitret}
\end{figure}

\begin{figure}
\begin{center}
\includegraphics[scale=0.2]{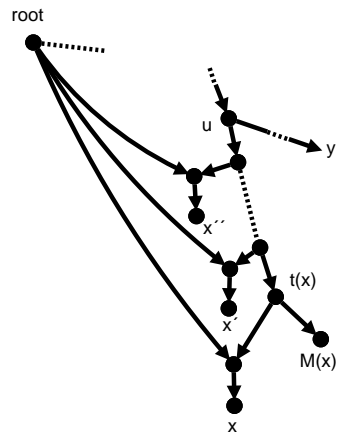}
\end{center}
\caption{This is case (b) in the proof of Theorem \ref{thm:6rA}. Here $u$ is a tree-node and there is a taxon $y \neq M(x)$ reachable by
a directed tree-path from $u$. Each intermediate node on the path from $u$ to $M(x)$ is the tail of a reticulation edge that feeds
into (the parent reticulation of) an SBR; the other reticulation edge is attached to the root. All these SBRs $x$, $x'$,...,$x''$ are such that $M(x)=M(x')=...=M(x'')$.} 
\label{fig:hittaxon}
\end{figure}

\begin{figure}
\begin{center}
\includegraphics[scale=0.2]{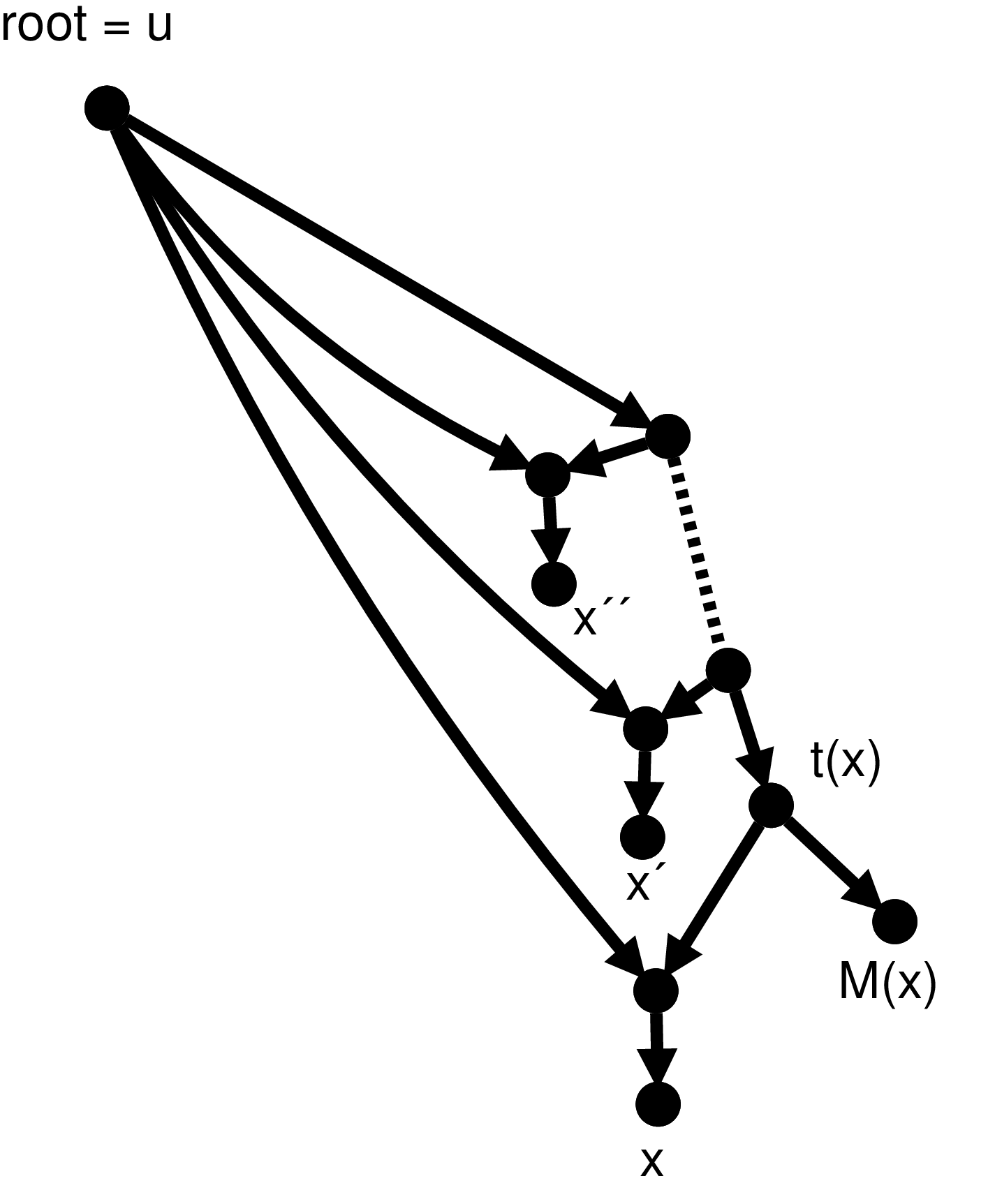}
\end{center}
\caption{This is case (c) in the proof of Theorem \ref{thm:6rA}. Here $u$ is the root and $\cX = \{x,x',\ldots,x''\} \cup \{M(x)\}$. Each intermediate node on the path from $u$ to $M(x)$ is the tail of a reticulation edge that feeds
into (the parent reticulation of) an SBR; the other reticulation edge is attached to the root. All these SBRs $x$, $x'$,...,$x''$ are such that $M(x)=M(x')=...=M(x'')$.} 
\label{fig:hitroot}
\end{figure}

\begin{lemma}
\label{lem:3r}
Let $\cC$ be an ST-collapsed set of clusters on $\cX$ such that $r(\cC) \geq 1$. Then 
$\cC$ contains at most $3 \cdot  r(\cC)$ terminals.
\end{lemma}
\begin{proof}
Given that $\cC$ is ST-collapsed, we know that there exists
a binary network $N$ with $r(N)=r(\cC)$ that represents $\cC$ such that $N$
can be obtained by performing the leaf-hanging operation to some $r(\cC)$-reticulation generator \cite{fptclusters}. The 1-reticulation and 2-reticulation generators are shown in Figure \ref{fig:rgen}.
Recall that the edges of a generator are called \emph{edge sides} and that the
nodes of a generator with indegree-2 and outdegree-0 are called \emph{node sides}. 
For a generator $G$ we let $I(G)$ be any maximum size subset of edge sides such that,
for every two sides $s \neq s'$ in $I(G)$, there is no directed path from the
head of $s$ to the tail of $s'$ such that all nodes on the path (including the head of $s$ and the tail of $s'$) are tree-nodes. We let $R(G)$ be the set of node sides of $G$, and we define $t(G)$ as $|R(G)| + |I(G)|$. We
define $t(r)$ (where $r \geq 1$) as the maximum value of $t(G)$ ranging over
all $r$-reticulation generators $G$. Observe that, if $r(\cC)=r$, then $t(r)$ is an upper
bound on the number of terminals in $\cC$. This follows because there can be at most one
terminal per edge side and, more generally,  it is not possible to place two terminals $x \neq y$ on the sides of the generator such that a directed tree-path exists from the parent of $x$ to the parent of $y$, because then $x {\contains} y$. To prove the lemma we will show that $t(r) \leq 3r$
for $r \geq 1$.

We will prove this by induction. The base case $r=1$ is straightforward. There
can be at most three terminals placed on the 1-reticulation generator: one on the node side
and one on each of the two edges whose head is the node side, see Figure \ref{fig:rgen}. (The cluster set
$\cC = \{\{a,b\}, \{b,c\}, \{d,c\}, \{e,d\}\}$, for which $r(\cC)=1$, shows that three terminals is actually possible).


Observe that, by acyclicity, every generator has at least
one node side. Furthermore, if we (1) delete a node side $s$ from an $r$-reticulation generator, (2) delete all leaves that are created (nodes with indegree-1 and outdegree-0) and (3) suppress all nodes with indegree and outdegree both equal to 1, we obtain an $(r-1)$-reticulation generator. For example, observe how applying steps (1)-(3) to any
2-reticulation generator creates the unique 1-reticulation generator. 

Now, for the sake of contradiction assume that $r > 1$ is the smallest value such that $t(r) > 3r$. Let $G$ be an $r$-reticulation generator such that
$t(G) \geq 3r+1$. Let $I(G)$ be a subset of the edge sides, as defined above, such that $t(G) = |R(G)| + |I(G)|$. Locate an arbitrary node side $s$ of $G$. We will
show that deleting $s$ and applying steps (1)-(3) as described above will create an $(r-1)$-generator $G'$ such that $t(G') \geq t(G) - 3$, yielding
a contradiction on the assumed minimality of $r$.  There are several cases, conditional on the positions of the tails of the two edges that enter $s$.\\
\\
In \textbf{case (i)} the two tails
are distinct and both have indegree-2 and outdegree-1. In this case, by the maximality of $I(G)$, both the edges entering $s$ will be in $I(G)$.
Hence, $|I(G')| = |I(G)|-2$ and $|R(G')| = |R(G)|+1$, so $t(G')=t(G)-1$.\\
\\
In \textbf{case (ii)} the
two tails are distinct and both have indegree-1 and outdegree-2. Clearly $|R(G')| = |R(G)|-1$. Let $u$ be the tail of the first edge that enters $s$. Let $p(u)$ 
be the parent of $u$ and $c(u)$ the child of $u$ not equal to $s$. The critical observation is that at most one of $(u,c(u))$ and $(p(u),u)$ is in $I(G)$. So
deleting $s$ will delete the edge $(u,s)$ from $I(G)$, if it is there, and if one of $(u,c(u))$ and $(p(u),u)$ is in $I(G)$ then this can be deleted and
replaced by the new edge $(p(u),c(u))$. The
same analysis holds for the second edge $(v,s)$ entering $s$. Hence $|I(G')| \geq |I(G)| - 2$, and this completes this case. Note that the analysis still holds
if (wlog) $p(v)=u$, because then at most one of the three edges $(p(u),u), (u,v), (v,c(v))$ will be in $I(G)$, and if this occurs this edge can be replaced
by the new edge $(p(u),c(v))$.\\
\\
In \textbf{case (iii)} both tails are distinct, one tail has indegree-2 and outdegree-1 and the other tail has indegree-1 and outdegree-2.
Then, by combining the insights from the first two cases, $|R(G')| = |R(G)|$ and $|I(G')| \geq |I(G)| - 2$.\\
\\
Finally, in \textbf{case (iv)} the two tails are the same vertex $u$ i.e. $s$ is the head
of a multi-edge. Note that at most two of the three edges $(p(u),u), (u,s), (u,s)$ can be in $I(G)$.
Now, suppose $p(u)$ has indegree-2 and outdegree-1. Then $|R(G')| = |R(G)|$ and $|I(G')| \geq |I(G)| - 2$. 
In the case that $p(u)$ has indegree-1 and outdegree-2 we have that $|R(G')|=|R(G)|-1$, and let $p'$ be the parent of $p(u)$ and $c'$ be
the child of $p(u)$ not equal to $u$. Again, at most one of the two edges $(p', p(u))$ and $(p(u),c')$ will be in $I(G)$, and
if necessary this can be replaced by the new edge $(p',c')$. So $|I(G')| \geq |I(G)|-2$.
\end{proof}

\begin{figure}
\begin{center}
\includegraphics[scale=0.2]{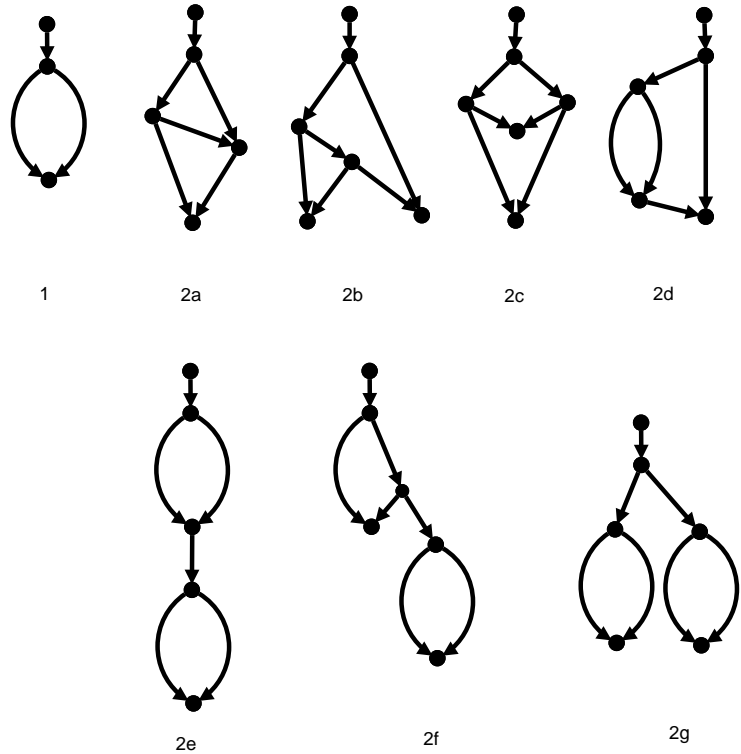}
\caption{The single 1-reticulation generator and the seven 2-reticulation generators.}
\label{fig:rgen}
\end{center}
\end{figure}

\begin{lemma}
\label{lem:6rB}
Let $\cC = Cl(\cT)$ be a set of clusters on $\cX$, where $\cT = \{T_1,T_2\}$ is a set of two binary trees on $\cX$ with no
non-trivial common subtrees, and $r(\cC) \geq 1$. Then there exists $\cX' \subseteq \cX$ such that (1) $|\cX'|  \leq 6 \cdot r(\cC)$
and (2) there exists $x \in \cX'$ such that $r(\cC \setminus \{x\}) = r(\cC) - 1$. Furthermore such a set $\cX'$ can be computed
in polynomial time.
\end{lemma}
\begin{proof}
If situation (1) from the statement of Theorem \ref{thm:6rA} holds then we can simply take $\cX' = \{a,b,c\}$, so $|\cX'|=3 \leq 6 \cdot r(\cC)$. Otherwise we are
in situation (2) and we take $\cX'$ to be the
union of all terminals of $\cC$ plus all taxa that appear in some size-2 cluster of $\cC$ with some terminal. From Lemma \ref{lem:3r} we know that
there are at most $3 \cdot r(\cC)$ terminals in $\cC$. Observe that each such terminal can be in at most one size-2 cluster, because otherwise
there would exist two incompatible size-2 clusters i.e. situation (1) of Theorem \ref{thm:6rA} would hold. Hence there can be at most
$3 \cdot r(\cC)$ non-terminals that appear in size-2 clusters with terminals, from which $|\cX'| \leq 6 \cdot r(\cC)$ follows.
Given the characterisation described in Observation \ref{obs:termequiv}, and the fact that $|\cC| \leq 4(n-1)$, it is easy to see that the set $\cX'$ can be computed in (low-order) polynomial time.
\end{proof}

\begin{theorem}
\label{thm:6rC}
Let $\cC = Cl(\cT)$ be a set of clusters on $\cX$, where $\cT = \{T_1,T_2\}$ is a set of two binary trees on $\cX$. Then $r(\cC) = h(\cT)$ can be computed
in time $(6r)^r \cdot poly(n)$ time, where $r = r(\cC)$ and $n = |\cX|$.
\end{theorem}
\begin{proof}
The algorithm is simple. We repeat the following steps until a compatible cluster set (i.e. a set of clusters that can be represented by a tree) is created:
(1) collapse any common subtrees into single taxa (adjusting $\cC$ as necessary), this can be done in (low-order) polynomial time, (2) construct the set $\cX'$ described in Lemma \ref{lem:6rB} and then (3) ``guess'' an element $x \in \cX'$ such that $r( \cC \setminus \{x\} ) = r( \cC ) - 1$. 

At each iteration the guessing can be simulated simply by trying all (at most) $6r$ elements
in $\cX'$. (Note that the $\cX'$ sets that arise will never have more than $6r$ elements because removing taxa from a cluster set cannot raise the reticulation
number of the cluster set). If we traverse this search tree in breadth-first fashion and stop as soon as we have created a compatible cluster set then the
depth of the search tree will equal $r = r(\cC)$, requiring at most $\Theta( (6r)^r )$ guesses in total.
\end{proof}

\section{Conclusion}

We have presented a new fixed parameter tractable algorithm for computing the hybridization number of two binary phylogenetic trees. The algorithm is unusual in the sense that it attacks the problem indirectly: it works within the softwired clusters model (which does not require the full topology of the input trees to be preserved) and links the optima together using the unification results in \cite{twotrees,elusiveness}. We hope that this will stimulate new insights into the underlying combinatorial structure
of the hybridization number problem.

\bibliography{6r}

\begin{thebibliography}{1}

\bibitem{fptclusters}
S.~M. Kelk and C.~Scornavacca.
\newblock Constructing minimal phylogenetic networks from softwired clusters is
  fixed parameter tractable, 2011.
\newblock Submitted, preprint available at
  \url{http://arxiv.org/abs/1108.3653}.

\bibitem{elusiveness}
S.~M. Kelk, C.~Scornavacca, and L.~J.~J. van Iersel.
\newblock On the elusiveness of clusters, 2011.
\newblock To appear in \emph{IEEE/ACM Transactions on Computational Biology and
  Bioinformatics}, preprint available at \url{http://arxiv.org/abs/1103.1834}.

\bibitem{twotrees}
L.~J.~J. van Iersel and S.~M. Kelk.
\newblock When two trees go to war.
\newblock {\em Journal of Theoretical Biology}, 269(1):245--255, 2011.

\end{thebibliography}

\end{document}